\documentclass[11pt]{article}

\usepackage{mathtools, amssymb, amsmath, amsfonts, amsthm, dsfont, rotating}
\usepackage{tikz, soul, titlesec}
\usepackage[toc,page]{appendix}
\usepackage[numbers]{natbib}
\usepackage[libertine]{newtxmath}
\usepackage{newtxtext}

\usetikzlibrary{arrows,patterns,positioning,calc,backgrounds,snakes}
\usetikzlibrary{decorations.pathmorphing,shapes}
\newcommand{\tikzmark}[1]{\tikz[overlay,remember picture] \node (#1) {};}

\newcounter{sarrow}

\usepackage[margin=1in]{geometry}
\usepackage{hyperref}
\hypersetup{
    colorlinks=true, % make the links colored
    citecolor=violet,
    linkcolor=blue, % color TOC links in blue
    urlcolor=red, % color URLs in red
    linktoc=all % 'all' will create links for everything in the TOC
}

\theoremstyle{plain}
\newtheorem{theorem}{Theorem}[section]
\newtheorem{lemma}[theorem]{Lemma}
\newtheorem{claim}[theorem]{Claim}

\newtheorem{fact}[theorem]{Fact}
\newtheorem{corollary}[theorem]{Corollary}

\newtheorem*{lemma*}{Lemma}
\newtheorem*{claim*}{Claim}
\newtheorem*{proposition*}{Proposition}
\newtheorem*{fact*}{Fact}
\newtheorem*{corollary*}{Corollary}
\newtheorem*{hint*}{Hint}

\theoremstyle{definition}
\newtheorem{definition}[theorem]{Definition}
\newtheorem{remark}[theorem]{Remark}
\newtheorem{notation}[theorem]{Notation}
\newtheorem{example}[theorem]{Example}

\newtheorem{problem}[theorem]{Problem}
\newtheorem{exercise}[theorem]{Exercise}

\newtheorem*{theorem*}{Theorem}
\newtheorem*{definition*}{Definition}
\newtheorem*{remark*}{Remark}
\newtheorem*{notation*}{Notation}
\newtheorem*{example*}{Example}
\newtheorem*{examples*}{Examples}
\newtheorem*{question*}{Question}
\newtheorem*{problem*}{Problem}
\newtheorem*{solution*}{Solution}
\newtheorem*{intuition*}{Intuition}
\newtheorem*{idea*}{Idea}
\newtheorem*{conjecture*}{Conjecture}

\newcommand{\btheorem}{\begin{theorem}}
\newcommand{\etheorem}{\end{theorem}}
\newcommand{\bproblem}{\begin{problem}}
\newcommand{\eproblem}{\end{problem}}
\newcommand{\bfact}{\begin{fact}}
\newcommand{\efact}{\end{fact}}
\newcommand{\bexercise}{\begin{exercise}}
\newcommand{\eexercise}{\end{exercise}}
\newcommand{\bclaim}{\begin{claim}}
\newcommand{\eclaim}{\end{claim}}
\newcommand{\bcorollary}{\begin{corollary}}
\newcommand{\ecorollary}{\end{corollary}}
\newcommand{\bnotation}{\begin{notation}}
\newcommand{\enotation}{\end{notation}}
\newcommand{\bremark}{\begin{remark}}
\newcommand{\eremark}{\end{remark}}
\newcommand{\blemma}{\begin{lemma}}
\newcommand{\elemma}{\end{lemma}}
\newcommand{\bexample}{\begin{example}}
\newcommand{\eexample}{\end{example}}
\newcommand{\bdefinition}{\begin{definition}}
\newcommand{\edefinition}{\end{definition}}
\newcommand{\bproof}{\begin{proof}}
\newcommand{\eproof}{\end{proof}}

\newcommand\whitespace{{\color{white}.}}

\newcommand{\bitem}{\whitespace\begin{itemize}} % Hack to fix alignment issues
\newcommand{\eitem}{\end{itemize}}
\newcommand{\be}{\whitespace\begin{enumerate}} % Hack to fix alignment issues
\newcommand{\ee}{\end{enumerate}}
\newcommand{\bi}{\whitespace\begin{itemize}} % Hack to fix alignment issues
\newcommand{\ei}{\end{itemize}}

\def\bal#1\eal{\begin{align*}#1\end{align*}}
\def\beq#1\eeq{\begin{equation}#1\end{equation}}
\def\bitem#1\eitem{\begin{itemize*}#1\end{itemize*}}
\def\bcomment#1\ecomment{\begin{comment}#1\end{comment}}

\newcommand{\ignore}[1]{}

\newcommand{\lc} {\lceil}
\newcommand{\rc} {\rceil}

\newcommand{\ceil}[1]{\lc {#1}\rc}

\newcommand{\ab}[1]{\langle#1\rangle}

\newcommand{\defeq}{\vcentcolon=}

\newcommand{\id}[1]{^{(#1)}}

\renewcommand{\Pr}{\mathop{\bf Pr\/}}
\newcommand{\E}{\mathop{\mathbb{E}\/}}

\DeclareMathOperator*\poly{poly}

\DeclareMathOperator*\Dec{Dec}

\DeclareMathOperator*\LCS{LCS}

\DeclareMathOperator*\Binomial{Binomial}
\DeclareMathOperator*\Poisson{Poisson}

\begin{document}

\title{Polynomial time decodable codes for the binary deletion channel\thanks{A conference paper on this work was presented at the 21st International Workshop on Randomization and Computation (RANDOM'2017), 16-18 August 2017.}}
\author{Venkatesan Guruswami\thanks{Computer Science Department, Carnegie Mellon University, Pittsburgh, PA 15213, USA. Email: {\tt venkatg@cs.cmu.edu}. Research supported in part by NSF grants CCF-1422045 and CCF-1814603.} \and Ray Li\thanks{Department of Computer Science, Stanford University, Stanford, CA 94305, USA. Email: {\tt rayyli@cs.stanford.edu}. Work done while at Carnegie Mellon University under support from an REU supplement to NSF CCF-1422045.}}
\date{Carnegie Mellon University \\ Pittsburgh, PA 15213}

\maketitle
\thispagestyle{empty}

\begin{abstract}
In the random deletion channel, each bit is deleted independently with probability $p$.  For the random deletion channel, the \emph{existence} of codes of rate $(1-p)/9$, and thus bounded away from $0$ for any $p < 1$, has been known. We give an explicit construction with polynomial time encoding and deletion correction algorithms with rate $c_0 (1-p)$ for an absolute constant $c_0 > 0$.
 \end{abstract}

%%%%%%%%%%%%%%%%%%%%%%%%%%%%%%%%%%%%%%%%%%%%%%%%%%%%%%%%%%%%%%%%%%%%%%%%%%%%%%%%%%%%%%%%%%%%%%%%%%%%%%%%%%%%%%%%%%%%%%%%%%%%%%%%%%%%%%%%%%%%%%%%%%%%%%%%%%%%%%%%
%%%%%%%%%%%%%%%%%%%%%%%%%%%%%%%%%%%%%%%%%%%%%%%%%%%%%%%%%%%%%%%%%%%%%%%%%%%%%%%%%%%%%%%%%%%%%%%%%%%%%%%%%%%%%%%%%%%%%%%%%%%%%%%%%%%%%%%%%%%%%%%%%%%%%%%%%%%%%%%%
%\newpage

\section{Introduction}
We consider the problem of designing explicit and efficient error-correcting codes for reliable communication on the binary deletion channel.
By \emph{explicit}, we mean codes such that we can produce a description of its encoding and decoding functions in time polynomial in the blocklength of the code.
By \emph{efficient}, we mean codes that are encodable and decodable in time polynomial in the blocklength of the code.
The \emph{binary deletion channel} (BDC) \emph{deletes} each transmitted bit independently with probability $p$, for some $p \in (0,1)$ which we call the \emph{deletion probability}. 
Crucially, the location of the deleted bits are \emph{not} known at the decoder, who receives a \emph{subsequence} of the original transmitted sequence. 
The loss of synchronization in symbol locations makes the noise model of deletions challenging to cope with. As one indication of this, we still do not know the channel capacity of the binary deletion channel. 
This is in sharp contrast with the noise model of bit erasures, where each bit is independently replaced by a '?' with probability $p$ (the binary erasure channel (BEC)), or of bit errors, where each bit is flipped independently with probability $p$ (the binary symmetric channel (BSC)). 
The capacity of the BEC and BSC equal $1-p$ and $1-h(p)$ respectively, and we know codes of polynomial complexity with rate approaching the capacity in each case. 

The capacity of the binary deletion channel is clearly at most $1-p$, the capacity of the simpler binary erasure channel.  
Diggavi and Grossglauser
\cite{DG2001} establish that the capacity of the
deletion channel for $p\le \frac{1}{2}$ is at least $1-h(p)$.
Kalai, Mitzenmacher, and Sudan~\cite{KMS2010} proved this lower
bound is tight as $p\to0$.
Works by Kirsch and Drinea \cite{KD2010}, and Mercier, Tarokh, and Labeau \cite{MercierTL2012} provide improved, computed lower bounds for capacity, the first for $p\le 0.35$, and the later for both small and large values of $p$.
Kanoria and Montanari \cite{KanoriaM2013} determined a series expansion that can be used to determine the capacity for $p\to 0$. 
Turning to large $p$, Rahmati and Duman \cite{RahmatiD15} prove that the capacity is at most $0.4143(1-p)$ for $p\ge 0.65$. 
Recently, Cheraghchi~\cite{Cheraghchi18} gave the first nontrivial BDC capacity upper bound outside the limiting case $p\to 0$ without computer assistance.
He proved that the capacity of the deletion channel is at most $(1-p)\log\varphi$ for $p\ge 1/2$, where $\varphi=\frac{1+\sqrt{5}}{2}$ is the golden ratio, and, under the conjecture that the capacity of the BDC is convex \cite{Dalai11}, at most $1-p\log(4/\varphi)$ for $p<1/2$.
Drinea and Mitzenmacher~\cite{DrineaM06, DrineaM07} proved that the capacity
of the BDC is at least $(1-p)/9$, which is within a
constant factor of the upper bound. In particular, the capacity is positive for every $p < 1$, which is perhaps surprising. The asymptotic behavior of the capacity of the BDC at both extremes of $p \to 0$ and $p \to 1$ is thus known.

Recently, there has been good progress on codes for \emph{adversarial} deletions, including explicit and efficiently decodable constructions. 
Here the model is that the channel can delete an \emph{arbitrary} subset of $pn$ bits in the $n$-bit codeword, after seeing the entire codeword. We again think of $p \in (0,1)$ as fixed and $n \to \infty$.
By standard concentration bounds, a code capable of correcting $pn$ worst-case deletions can also correct deletions caused by a BDC with deletion probability $(p-\varepsilon)$ with high probability, so one can infer results for the BDC from some results for worst-case deletions. For small $p$, Guruswami and Wang \cite{GW17} constructed binary codes of rate $1-O(\sqrt{p})$ to efficiently correct a $p$ fraction worst-case deletions. 
Recent work on the document exchange problem \cite{ChengJLW18,Haeupler18} provided an improved efficient construction, achieving rate $1-O(p\log^2(1/p))$.
Hence, we have explicit codes of rate approaching $1$ for the BDC when $p \to 0$. For larger $p$, Kash et
al. \cite{KashMTU11} proved that randomly chosen codes of small enough
rate $R > 0$ can correctly decode against $pn$ adversarial deletions
when $p\le 0.17$. Even non-explicitly, this remained the best
achievability result in terms of correctable deletion fraction until the recent work of 
Bukh, Guruswami, and H{\aa}stad~\cite{BGH17} who constructed codes of positive rate 
efficiently decodable against $pn$ adversarial deletions for any $p<\sqrt2-1$. For adversarial deletions, it is impossible to correct a deletion fraction of $1/2$, whereas the capacity of the BDC is positive for all $p < 1$. So solving the problem for the much harder worst-case deletions is not a viable approach to construct positive rate codes for the BDC for $p>1/2$.

To the best of our knowledge, explicit efficiently decodable code constructions were not available for the binary deletion channel for arbitrary $p < 1$. We present such a construction in this work.  Our rate is worse than the $(1-p)/9$ achieved non-explicitly, but has asymptotically the same dependence on $p$ for $p \to 1$.
We acknowledge that our code, while explicit and efficient, is not intended to be practical. 
\begin{theorem}
  \label{thm:random}
  Let $p \in (0,1)$. There is an explicit family of binary codes that
  (1) has rate $(1-p)/120$, (2) is constructible in time $\poly(N)$, (3) encodable in time $O(N)$, and (3) decodable with high probability on the binary deletion channel with deletion probability $p$ in time $O(N^2)$. Here $N$ is the block length of the code.
\end{theorem}
\subsection{Some other related work}

One work that considers efficient recovery against random deletions is by Yazdi and Dolecek \cite{YazdiD2014}.
    In their setting, two parties Alice and Bob are connected by a two-way communication channel.
    Alice has a string $X$, Bob has string $Y$ obtained by passing $X$ through a binary deletion channel with deletion probability $p\ll 1$, and Bob must recover $X$.
    They produce a polynomial-time synchronization scheme that transmits a total of $O(pn\log(1/p))$ bits and allows Bob to recover $X$ with probability exponentially approaching 1.
    A follow-up work \cite{SalaSBD16} considers a similar model with larger alphabets, nonuniform distributions of alphabet symbols, and tolerating both insertions and deletions.

For other models of random synchronization errors, Kirsch and Drinea \cite{KD2010} prove information capacity lower bounds for channels with i.i.d deletions and duplications. Building on a previous work for the deletion channel \cite{FertonaniD10}, Fertonani et al. \cite{FDE2011} prove capacity bounds for binary channels with i.i.d insertions, deletions, and substitutions.
Haeupler and Mitzenmacher \cite{HaeuplerM14} show that the capacity of sending $k$ copies of a codeword through a BDC with deletion probability $p$ is $1-\alpha H(p^k) - O(p^k)$ for an explicit constant $\alpha$.

For deletion channels over non-binary alphabets, Rahmati and Duman \cite{RahmatiD15} prove a capacity upper bound of $C_2(p)+(1-p)\log(|\Sigma|/2)$, where $C_2(p)$ denotes the capacity of the binary deletion channel with deletion probability $p$, when the alphabet size $|\Sigma|$ is even.
In particular, using the best known bound for $C_2(p)$ of $C_2(p)\le0.4143(1-p)$, the upper bound is $(1-p)(\log|\Sigma|-0.5857)$.
 
In \cite{GuruswamiL2017}, the authors of this paper consider the model of \emph{oblivious} deletions, which is in between the BDC and adversarial deletions in power. Here, the channel can delete any $pn$ bits of the codeword, but must do so without knowledge of the codeword. In this model, they prove the \emph{existence} of codes of positive rate for correcting any fraction $p < 1$ of oblivious deletions.

\vspace{-1ex}
\subsection{Our construction approach}
\vspace{-1ex}
Our construction concatenates a high rate outer code over a large alphabet that is efficiently decodable against a small fraction of \emph{adversarial} insertions and deletions, with a good inner binary code.
For the outer code, we can use the recent construction of \cite{HaeuplerS17}.
To construct the inner code, we first choose a binary code correcting a small fraction of adversarial deletions.
By concentration bounds, duplicating bits of a codeword in a disciplined manner is effective against the random deletion channel, so we, for some constant $B$, duplicate every bit of the binary code $B/(1-p)$ times.
We further ensure our initial binary code has only runs of length 1 and 2 to maximize the effectiveness of duplication.
We add small buffers of 0s between inner codewords to facilitate decoding.

One might wonder whether it would be possible to use Drinea and Mitzenmacher's existential result \cite{DrineaM06,DrineaM07} of a $(1-p)/9$ capacity lower bound as a black box inner code to achieve a better rate together with efficient decodability.
We discuss this approach in \S\ref{sec:5-2} and elaborate on what makes such a construction difficult to implement.

\vspace{-1ex}  
\section{Preliminaries}
\label{sec:3}
\vspace{-1ex}
  \textbf{General Notation.}
  Throughout the paper, $\log x$ refers to the base-2 logarithm.

  We use interval notation $[a,b] = \{a,a+1,\dots,b\}$ to denote intervals of integers, and we use $[a] = [1,a] = \{1,2,\dots,a\}$.

  Let $\Binomial(n,p)$ denote the Binomial distribution.

  \textbf{Words.}
  A \emph{word} is a sequence of symbols from some \emph{alphabet}.
  We denote explicit words using angle brackets, like $\ab{01011}$. 
  We denote string concatenation of two words $w$ and $w'$ with $ww'$.
  We denote $w^k=ww\cdots w$ where there are $k$ concatenated copies of $w$.

  A \emph{subsequence} of a word $w$ is a word obtained by removing some (possibly none) of the symbols in $w$.

  Let $\Delta_{i/d}(w_1,w_2)$ denote the \emph{insertion/deletion distance} between $w_1$ and $w_2$, i.e. the minimum number of insertions and deletions needed to transform $w_1$ into $w_2$.
  By a lemma due to Levenshtein \cite{Levenshtein1966}, this is equal to $|w_1|+|w_2| - 2\LCS(w_1,w_2)$, where $\LCS$ denotes the length of the longest common subsequence.

    Define a \textit{run} of a word $w$ to be a maximal single-symbol subword. That is, a subword $w'$ in $w$ consisting of a single symbol such that any longer subword containing $w'$ has at least two different symbols.
    Note the runs of a word partition the word. For example, $110001$ has 3 runs: one run of 0s and two runs of 1s.

    We say that $c\in\{0,1\}^m$ and $c'\in\{0,1\}^m$ are \emph{confusable under $\delta m$ deletions} if it is possible to apply $\delta m$ deletions to $c$ and $c'$ and obtain the same result. If $\delta$ is understood, we simply say $c$ and $c'$ are \emph{confusable}.

  \textbf{Concentration Bounds.}
  We use the following forms of Chernoff bound (see, e.g., Theorems 4.4 and 4.5 of \cite{MitzenmacherU05}):
  \begin{lemma}
    \label{lem:chernoff}
    Let $A_1,\dots,A_n$ be i.i.d random variables taking values in $[0,1]$.
    Let $A = \sum_{i=1}^n A_i$ and $\delta\in[0,1]$.
    Then
    \begin{align}
      \Pr[A\le (1-\delta)\E[A]] \ \le \ \exp\left(-\delta^2\E[A]/2\right).
      \label{eq:chernoff}
    \end{align}
    Furthermore, 
    \begin{equation}
      \Pr[A\ge (1+\delta)\E[A]] 
      \ \le \ \left( \frac{e^\delta}{(1+\delta)^{1+\delta}}\right)^{\E[A]}.
      \label{eq:chernoff-2}
    \end{equation}
  \end{lemma}

  We also have the following corollary, whose proof is in Appendix~\ref{app:A}.
  \begin{lemma}
    \label{lem:chernoff-2}
    Let $0<\alpha<\beta$.
    Let $A_1,\dots,A_n$ be independent random variables taking values in $[0,\beta]$ such that, for all $i$, $\E[A_i]\le \alpha$.
    For $\gamma\in[\alpha, 2\alpha]$, we have
    \begin{equation}
      \Pr\left[\sum_{i=1}^n A_i \ge n\gamma\right] \le \exp\left( -\frac{(\gamma - \alpha)^2n}{3\alpha\beta} \right).
    \end{equation}
  \end{lemma}

%%%%%%%%%%%%%%%%%%%%%%%%%%%%%%%%%%%%%%%%%%%%%%%%%%%%%%%%%%%%%%%%%%%%%%%%%%%%%%%%%%%%%%%%%%%%%%%%%%%%%%%%%%%%%%%%%%%%%%%%%%%%%%%%%%%%%%%%%%%%%%%%%%%%%%%%%%%%%%%%
%%%%%%%%%%%%%%%%%%%%%%%%%%%%%%%%%%%%%%%%%%%%%%%%%%%%%%%%%%%%%%%%%%%%%%%%%%%%%%%%%%%%%%%%%%%%%%%%%%%%%%%%%%%%%%%%%%%%%%%%%%%%%%%%%%%%%%%%%%%%%%%%%%%%%%%%%%%%%%%%
\section{Efficient codes for random deletions with $p$ approaching 1}
\label{sec:5}

\subsection{Construction}
\label{sec:5-1}

We present a family of constant rate codes that decodes with high probability on a binary deletion channel with deletion fraction $p$ (BDC$_p$). These codes have rate $c_0(1-p)$ for an absolute positive constant $c_0$, which is within a constant of the upper bound $(1-p)$, which even holds for the erasure channel. By Drinea and Mitzenmacher \cite{DrineaM06} the maximum known rate of a non-efficiently correctable binary deletion channel code is $(1-p)/9$.

We begin by borrowing a result from \cite{GW17}.
\begin{lemma}[Corollary of Lemma 2.3 of \cite{GW17}]
  Let $0<\delta<\frac{1}{2}$. For every binary string $c\in\{0,1\}^m$, there are at most $\delta m\binom{m}{(1-\delta)m}^2$ strings $c'\in\{0,1\}^m$ such that $c$ and $c'$ are confusable under $\delta m$ deletions.
\end{lemma}
The next lemma gives codes against a small fraction of adversarial deletions with an additional run-length constraint on the codewords.
\begin{lemma}
  \label{lem:binary-insdel-small-run-existence-1}
  Let $\delta > 0$ and $m$ be a positive integer.
  There exists a length $m$ binary code of rate $\mathcal{R}=0.6942-2h(\delta)-O(\log(\delta m)/m)$ correcting a $\delta$ fraction of adversarial insertions and deletions such that each codeword contains only runs of size 1 and 2.
  Furthermore this code is constructible in time $\tilde O(2^{(0.6942+\mathcal{R})m})$.
\end{lemma}
\begin{proof}
  It is easy to show that the number of codewords with only runs of 1 and 2 is $F_m$, the $m$th Fibonacci number, and it is well known that $F_m = \varphi^m + o(1) \approx 2^{0.6942m}$ where $\varphi$ is the golden ratio.
  Now we construct the code by choosing it greedily.
  Each codeword is confusable with at most $\delta m\binom{m}{(1-\delta)m}^2$ other codewords, so the number of codewords we can choose is at least
  \begin{equation}
    \frac{2^{0.6942m}}{\delta m\binom{m}{(1-\delta)m}^2}\ = \ 2^{m(0.6942-2h(\delta)-O(\log(\delta m)/m))}.
  \end{equation}
  
  We can find all words of length $m$ whose run lengths are only 1 and 2 by recursion in time $O(F_m) = O(2^{0.6942m})$.
  Running the greedy algorithm, we need to, for at most $F_m\cdot 2^{\mathcal{R} m}$ pairs of such words, determine whether the pair is confusable (we only need to check confusability of a candidate word with words already added to the code).
  Checking confusability of two words under adversarial deletions reduces to checking whether the longest common subsequence is at least $(1-\delta)m$, which can be done in time $O(m^2)$.
This gives an overall runtime of $O(m^2\cdot F_m\cdot 2^{\mathcal{R} m}) = \tilde O(2^{(0.6942+\mathcal{R})m})$.
\end{proof}
\def\rinn{0.555}
\def\rin{0.554}
\def\deltain{0.0083}
\begin{corollary}
  \label{lem:inner}
  There exists a constant $m_0^*$ such that for all $m\ge m_0^*$, there exists a length $m$ binary code of rate $\mathcal{R}=\rinn$ correcting a $\delta=\deltain$ fraction of adversarial insertions and deletions such that each codeword contains runs of size 1 and 2 only and each codeword starts and ends with a 1.
  Furthermore this code is constructible in time $O(2^{1.25m})$.
\end{corollary}

Our construction utilizes the following result as a black box for efficiently coding against an arbitrary fraction of insertions and deletions with rate approaching capacity.
\begin{theorem}[Theorem 1.1 of \cite{HaeuplerS17}]
  For any $0\le \delta < 1$ and $\varepsilon > 0$, there exists a code $C$ over alphabet $\Sigma$, with $|\Sigma| = O_\varepsilon(1)$, with block length $n$, rate $1-\delta - \varepsilon$, and is efficiently decodable from $\delta n$ insertions and deletions. 
  The code can be constructed in time $\poly(n)$, encoded in time $O(n)$, and decoded in time $O(n^2)$.
  \label{thm:hs17}
\end{theorem}
We apply Theorem~\ref{thm:hs17} for small $\delta$, so we also could use the high rate binary code construction of \cite{GuruswamiL2016} as an outer code.

We now turn to our code construction for Theorem~\ref{thm:random}.

  \textbf{The code.}
  Throughout this construction, $n$, the length of the outer code, and $N$, the total length of the code, tend to $\infty$, and $B,p,\eta,\delta_{in},\delta_{out},\varepsilon_{out},\mathcal{R}_{in},|\Sigma|,m$ are absolute constants.
  Pick parameters $B = 60, B^* = 1.4\bar{3}B = 86, \eta = \frac{1}{1000}, \delta_{out} = \frac{1}{1000}$.
  Let $m_0 = \max(\alpha\log(1/\delta_{out})/\eta,m_0^*)$, where $\alpha$ is a sufficiently large constant and where $m_0^*$ is given by Corollary~\ref{lem:inner}.
  Let $\varepsilon_{out}\in(0,\frac{1}{1000})$ be small enough such that the alphabet $\Sigma$, given by Theorem~\ref{thm:hs17} with $\varepsilon=\varepsilon_{out}$ and $\delta=\delta_{out}$, satisfies $\frac{1}{\mathcal{R}_{in}}\log|\Sigma|\ge m_0$, and let $C_{out}$ be the corresponding code of length $n$.
  
  Let $C_{in}:|\Sigma|\to\{0,1\}^m$ be the code given by Corollary~\ref{lem:inner}, and let $m=\ceil{\frac{1}{\mathcal{R}_{in}}\log|\Sigma|}=O_\varepsilon(1)$ be the block length of the code, respectively.
  The code $C_{in}$ tolerates at least $\delta_{in}=\deltain$ fraction of adversarial insertions and deletions, and the rate is at least $\mathcal{R}_{in}=\rin$.\footnote{$\mathcal{R}_{in}$ and $\delta_{in}$ are given by Corollary~\ref{lem:inner}. $\mathcal{R}_{in}$ is a little smaller because of rounding $m$}
  Each codeword of $C_{in}$ has runs of length 1 and 2 only, and each codeword starts and ends with a 1. 
  This code is constructed greedily.

  Our code is a modified concatenated code.
  We encode our message as follows.
  \begin{itemize}
    \item \emph{Outer Code.} First, encode the message into the outer code, $C_{out}$, to obtain a word $c\id{out} = \sigma_1\dots\sigma_n$.
    \item \emph{Concatenation with Inner Code.} Encode each outer codeword symbol $\sigma_i\in\Sigma$ by the inner code $C_{in}$.
    \item \emph{Buffer.} Insert a buffer of $\ceil{\eta m}$ 0s between adjacent inner codewords.
      Let the resulting word be $c\id{cat}$.
      Let $c_i\id{in} = C_{in}(\sigma_i)$ denote the encoded inner codewords of $c\id{cat}$.
    \item \emph{Duplication.} After concatenating the codes and inserting the buffers, replace each bit (including bits in the buffers) with $\ceil{B/(1-p)}$ copies of itself to obtain a word of length $N \defeq (nm + (n-1) \ceil{\eta m})\cdot\ceil{B/(1-p)}$.
      Let the resulting word be $c$, and the corresponding inner codewords be $\{c\id{dup}_i\}$.
  \end{itemize}
  We emphasize that the same inner code can be used for all values of $n$ and $p$: the dependence of our code's parameters on $n$ comes from the outer code, and the dependence on $p$ comes from the duplication step.

  \textbf{Rate.}
  The rate of the outer code is $1-\delta_{out} - \varepsilon_{out}$, the rate of the inner code is $\mathcal{R}_{in}$, the buffer and duplications multiply the rate by no less than $\frac{1}{1+\eta+O(1/m)}$ and $(1-p)/(B+1)$, respectively.
  This gives a total rate that is greater than $(1-p)/120$.

  \textbf{Notation.} Let $s$ denote the received word after the codeword $c$ is passed through the deletion channel.
  Note that (i) every bit of $c$ can be identified with a bit in $c\id{cat}$, and (ii) each bit in the received word $s$ can be identified with a bit in $c$.
  Thus, we can define relations $f\id{dup}:c\id{cat}\to c$, and $f\id{del}:c\to s$ (that is, relations on the indices of the strings).
  These are not functions because some bits may be mapped to multiple (for $f\id{dup}$) or zero (for $f\id{del}$) bits.
  Specifically, $f\id{del}$ and $f\id{dup}$ are the inverses of total functions.
  In this way, composing these relations (i.e. composing their inverse functions) if necessary, we can speak about the \emph{image} and \emph{pre-image} of bits or subwords of one of $c\id{cat},c,$ and $s$ under these relations.
  For example, during the Duplication step of encoding, a bit $\ab{b_j}$ of $c\id{cat}$ is replaced with $\ceil{B/(1-p)}$ copies of itself, so the corresponding string $\ab{b_j}^{\ceil{B/(1-p)}}$ in $c$ forms the \emph{image} of $\ab{b_j}$ under $f\id{dup}$, and conversely the \emph{pre-image} of the duplicated string $\ab{b_j}^{\ceil{B/(1-p)}}$ is that bit $\ab{b_j}$.

  \textbf{Decoding algorithm.}
  \begin{itemize}
    \item \emph{Decoding Buffer.} First identify all runs of 0s in the received word with length at least $B\eta m / 2$. These are our \emph{decoding buffers} that divide the word into \emph{decoding windows}, which we identify with subwords of $s$.
    
    \item \emph{Deduplication.} Divide each decoding window into runs. 
      For each run, if it has strictly more than $B^*$ copies of a bit, replace it with as two copies of that bit, otherwise replace it with one copy.
      For example, $\ab{0}^{2B}$ is replaced with $\ab{00}$ while $\ab{0}^B$ is replaced with $\ab{0}$.
      For each decoding window, concatenate these runs of length 1 and 2 in their original order in the decoding window to produce a \emph{deduplicated} decoding window.
    \item \emph{Inner Decoding.} For each deduplicated decoding window, decode an outer symbol $\sigma \in \Sigma_{out}$ from each decoding window by running the brute force deletion correction algorithm for $C_{in}$.
      That is, for each deduplicated decoding window $s_*\id{in}$, find by brute force a codeword $c_*\id{in}$ in $C_{in}$ that such that $\Delta_{i/d}(c_*\id{in},s_*\id{in}) \le \delta_{in}m$.
      If $c_*\id{in}$ is not unique or does not exist, do not decode an outer symbol $\sigma$ from this decoding window.
      Concatenate the decoded symbols $\sigma$ in the order in which their corresponding decoding windows appear in the received word $s$ to obtain a word $s\id{out}$.
    \item \emph{Outer Decoding.} Decode the message $\mathfrak{m}$ from $s\id{out}$ using the decoding algorithm of $C_{out}$ in Theorem~\ref{thm:hs17}.
  \end{itemize}
    For purposes of analysis, label as $s_i\id{dup}$ the decoding window whose pre-image under $f\id{del}$ contains indices in $c_i\id{dup}$.
    If this decoding window is not unique (that is, the image of $c_i\id{dup}$ contains bits in multiple decoding windows), then assign $s_i\id{dup}$ arbitrarily.
    Note this labeling may mean some decoding windows are unlabeled, and also that some decoding windows may have multiple labels.
    In our analysis, we show both occurrences are rare.
    For a decoding window $s_i\id{dup}$, denote the result of $s_i\id{dup}$ after Deduplication to be $s_i\id{in}$.

  The following diagram depicts the encoding and decoding steps.
  The pair $(\{c_i\id{in}\}_i,c\id{cat})$ indicates that, at that step of encoding, we have produced the word $c\id{cat}$, and the sequence $\{c_i\id{in}\}_i$ are the ``inner codewords'' of $c\id{cat}$ (that is, the words in between what would be identified by the decoder as decoding buffers).
  The pair $(\{c_i\id{dup}\}_i,c)$ is used similarly.

  \begin{align*}
    \mathfrak{m}
    \xrightarrow{C_{out}}
      c\id{out}
    \xrightarrow{C_{in},Buf}
      &\left(\left\{c_i\id{in}\right\}_i , c\id{cat}\right)
    \xrightarrow{Dup}
    \left(\left\{c_i\id{dup}\right\}_i ,
             c\right)\tikzmark{a}
      \\\\\\
    &\,\tikzmark{b}
        s
    \xrightarrow{DeBuf}
      \left\{s_i\id{dup}\right\}_i
    \xrightarrow{DeDup}
       \left\{s_i\id{in}\right\}_i
    \xrightarrow{\Dec_{in}}
      s\id{out}
    \xrightarrow{\Dec_{out}}
      \mathfrak{m}
        \begin{tikzpicture}[overlay,remember picture,out=315,in=135,distance=2cm]
      \draw[->, red,shorten >=3pt,shorten <=3pt] (a.center) to node [midway, label=above:BDC] {} (b.center)  ;
        \end{tikzpicture}
    \label{random:}
  \end{align*}

  \textbf{Runtime.}
  The outer code is constructible in $\poly(n)$ time and the inner code is constructible in time $O(2^{1.25m}) = O_\varepsilon(1)$, which is a constant, so the total construction time is $\poly(N)$.

  Encoding in the outer code is linear time, each of the $n$ inner encodings is constant time, and adding the buffers and applying duplications each can be done in linear time.
  The overall encoding time is thus $O(N)$.

  The Buffer step of the decoding takes linear time.
  The Deduplication step of each inner codeword takes constant time, so the entire step takes linear time.
  For each inner codeword, Inner Decoding takes time $O(m^22^m) = O_\varepsilon(1)$ by brute force search over the $2^m$ possible codewords:
  checking each of the $2^m$ codewords is a longest common subsequence computation and thus takes time $O(m^2)$, giving a total decoding time of $O(m^22^m)$ for each inner codeword.
  We need to run this inner decoding $O(n)$ times, so the entire Inner Decoding step takes linear time.
  The Outer Decoding step takes $O(n^2)$ time by Theorem~\ref{thm:hs17}.
  Thus the total decoding time is $O(N^2)$.

  \textbf{Correctness.}
  Note that, if an inner codeword is decoded incorrectly, then one of the following holds.
  \begin{enumerate}
    \item (\emph{Spurious Buffer}) A spurious decoding buffer is identified in the corrupted codeword during the Buffer step.
    \item (\emph{Deleted Buffer}) A decoding buffer neighboring the codeword is deleted.
    \item (\emph{Inner Decoding Failure}) Running the Deduplication and Inner Decoding steps on $s_i\id{dup}$ computes the inner codeword incorrectly.
  \end{enumerate}
  We show that, with high probability, the number of occurrences of each of these events is small. 

  The last case is the most complicated, so we deal with it first, assuming the codeword contains no spurious decoding buffers and the neighboring decoding buffers are not deleted.
  In particular, we consider an $i$ such that our decoding window $s_i\id{dup}$ whose pre-image under $f\id{del}$ only contains bits in $c_i\id{dup}$ (because no deleted buffer) and no bits in the image of $c_i\id{dup}$ appear in any other decoding window (because no spurious buffer).

  Recall that the inner code $C_{in}$ can correct against $\delta_{in}=0.0083$ fraction of adversarial insertions and deletions.
  Suppose an inner codeword $c_i\id{in}=r_1\dots r_k\in C_{in}$ has $k$ runs $r_j$ each of length 1 or 2, so that $m/2\le k\le m$.
  \begin{definition}
    A subword of $\alpha$ identical bits in the received word $s$ is
    \begin{itemize}
      \item \emph{type-0} if $\alpha=0$
      \item \emph{type-1} if $\alpha\in[1,B^*]$
      \item \emph{type-2} if $\alpha\in[B^*+1,\infty)$. 
    \end{itemize}
    By abuse of notation, we say that a length 1 or 2 run $r_j$ of the inner codeword $c_i\id{in}$ has \emph{type-$t_j$} if the image of $r_j$ in $s$ under $f\id{del}\circ f\id{dup}$ forms a type-$t_j$ subword.\qedhere
  \end{definition}

  Let $t_1,\dots,t_k$ be the types of the runs $r_1,\dots,r_k$, respectively.
  The image of a run $r_j$ under $f\id{del}\circ f\id{dup}$ has length distributed as $Z\sim\Binomial(|r_j|\cdot\ceil{B/(1-p)}, 1-p)$.
  Let $\delta=0.4\bar{3}$ and $\delta'=0.426\dots$ be such that $B^* = (1+\delta)B = (1+\delta')(B+1)-1$, so that $(1+\delta')\E[Z] < B^*+1$.
  By the Chernoff bounds in Lemma~\ref{lem:chernoff}, the probability that a run $r_j$ of length 1 is type-2 is 
  \begin{equation}
    \Pr_{Z\sim\Binomial(\ceil{B/(1-p)},1-p)}[Z \ge B^* + 1] 
    \ \le \ \Pr_{Z}[Z > (1+\delta')\E[Z]] 
    \ \le \  \left(e^{\delta'}/(1+\delta')^{1+\delta'}\right)^B < 0.00817.
    \label{eq:random-cher-1}
  \end{equation}
  Similarly, the probability that a run $r_j$ of length-2 is type-1 is at most
  \begin{equation}
    \Pr_{Z\sim\Binomial(2\ceil{B/(1-p)},1-p)}[Z \le B^*] < e^{-((1-\delta)/2)^2B} < 0.0081.
    \label{eq:random-cher-2}
  \end{equation}
  The probability any run is type-0 is  at most $\Pr_{Z\sim\Binomial(\ceil{B/(1-p)},1-p)}[Z=0]<e^{-B}<10^{-10}$.

  We now have established that, for runs $r_j$ in $c_i\id{in}$, the probability that the number of bits in the image of $r_j$ in $s$ under $f\id{del}\circ f\id{dup}$ is ``incorrect'' (between 1 and $B^*$ for length 2 runs, and greater than $B^*$ for length 1 runs), is less than $0.0082$, which is less than $\delta_{in}$.
  If the only kinds of errors in the Local Decoding step were runs of $c$ of length 1 becoming runs of length 2 and runs of length 2 become runs of length 1, then we have that, by concentration bounds, with probability $1-2^{-\Omega(m)}$, the number of insertions deletions needed to transform $s_i\id{in}$ back into $c_i\id{in}$ is at most $\delta_{in}m$, in which case $s_i\id{in}$ gets decoded to the correct outer symbol using $C_{in}$.

  However, we must also account for the fact that some runs $r_j$ of $c_i\id{in}$ may become deleted completely after duplication and passing through the deletion channel. That is, the image of $r_j$ in $s$ under $f\id{del}\circ f\id{dup}$ is empty, or, in other words, $r_j$ is type-0. In this case the two neighboring runs $r_{j-1}$ and $r_{j+1}$ appear merged together in the Deduplication step of decoding.
  For example, if a run of 1s was deleted completely after duplication and deletion, its neighboring runs of 0s would be interpreted by the decoder as a single run.
  Fortunately, as we saw, the probability that a run is type-0 is extremely small ($< 10^{-10}$), and we show each type-0 run only increases $\Delta_{i/d}(c_i\id{in}, s_i\id{in})$ by a constant.
  We show this constant is at most 6.

  To be precise, let $Y_j$ be a random variable that is $0$ if $|r_j|=t_j$, 1 if $\{|r_j|, t_j\}=\{1,2\}$, and 6 if $t_j=0$. 
  We claim $\sum_{j=1}^kY_j$ is an upper bound on $\Delta_{i/d}(c_i\id{in}, s_i\id{in})$. 
  To see this, first note that if $t_j\neq0$ for all $i$, then the number of runs of $c_i\id{in}$ and $s_i\id{in}$ are equal, so we can transform $c_i\id{in}$ into $s_i\id{in}$ by adding a bit to each length-1 type-2 run of $c_i\id{in}$ and deleting a bit from each length-2 type-1 run of $s_i\id{in}$.
  
  Now, if some number, $\ell$, of the $t_j$ are $0$, then at most $2\ell$ of the runs in $c_i\id{in}$ become merged with some other run (or a neighboring decoding buffer) after duplication and deletion. 
  Each set of consecutive runs $r_j,r_{j+2},\dots,r_{j+2j'}$ that are merged after duplication and deletion is replaced with 1 or 2 copies of the corresponding bit.
  For example, if $r_1=\ab{11},r_2=\ab{0},r_3=\ab{11}$, and if after duplication and deletion, $2B$ bits remain in the image of each of $r_1$ and $r_3$, and $r_2$ is type-0, then the image of $r_1r_2r_3$ under $f\id{del}\circ f\id{dup}$ is $\ab{1}^{4B}$, which gets decoded as $\ab{11}$ in the Deduplication step because $\ab{1}^{4B}$ is type-2.
  To account for the type-0 runs in transforming $c_i\id{in}$ into $s_i\id{in}$, we (i) delete at most two bits from each of the $\ell$ type-0 runs in $c_i\id{in}$ and (ii) delete at most two bits for each of at most $2\ell$ merged runs in $c_i\id{in}$.
  The total number of additional insertions and deletions required to account for type-0 runs of $c$ is thus at most $6\ell$, so we need at most $6$ insertions and deletions to account for each type-0 run.

  Our analysis covers the case when some bits in the image of $c_i\id{in}$ under $f\id{del}\circ f\id{dup}$ are interpreted as part of a decoding buffer.
  Recall that inner codewords start and end with a 1, so that $r_1\in\{\ab{1},\ab{11}\}$ for every inner codeword.
  If, for example, $t_1=0$, that is, the image under $f\id{del}\circ f\id{dup}$ of the first run of 1s, $r_1$, is the empty string, then the bits of $r_2$ are interpreted as part of the decoding buffer.
  In this case too, our analysis tells us that the type-0 run $r_1$ increases $\Delta_{i/d}(c_i\id{in},s_i\id{in})$ by at most $6$.

  We conclude $\sum_{j=1}^kY_j$ is an upper bound for $\Delta_{i/d}(c_i\id{in},s_i\id{in})$.

  Note that if $r_j$ has length $1$, then by \eqref{eq:random-cher-1} we have
  \begin{equation}
    \E[Y_j] \ = \ 1\cdot\Pr[\text{$r_j$ is type-2}] + 6\cdot\Pr[\text{$r_j$ is type-0}]
      \ < \ 1\cdot0.00817 + 6\cdot 10^{-9}
      \ < \ 0.0082.
    \label{}
  \end{equation}
  Similarly, if $r_j$ has length $2$, then by \eqref{eq:random-cher-2} we have
  \begin{equation}
    \E[Y_j] \ = \ 1\cdot\Pr[\text{$r_j$ is type-1}] + 6\cdot\Pr[\text{$r_j$ is type-0}]
      \ < \ 1\cdot0.0081 + 6\cdot 10^{-9}
      \ < \ 0.0082.
    \label{}
  \end{equation}
  Thus $\E[Y_j] < 0.0082$ for all $i$.
  We know the word $s_i\id{in}$ is decoded incorrectly (i.e. is not decoded as $\sigma_i$) in the Inner Decoding step only if $\Delta_{i/d}(c_i\id{in},s_i\id{in}) > \delta_{in}m$.
  The $Y_j$ are independent, so Lemma~\ref{lem:chernoff-2} gives
  \begin{align}
    \Pr[\text{$s_i\id{in}$ decoded incorrectly}]
    \ &\le \ \Pr[Y_1+Y_2+\cdots+Y_k\ge \delta_{in}m] \nonumber\\
      \ &\le \ \Pr[Y_1+Y_2+\cdots+Y_k\ge \delta_{in}k] \nonumber\\
      \ &\le \ \exp\left( -\frac{(\delta_{in} - 0.0082)^2k}{3\cdot 6\cdot \delta_{in}} \right) \nonumber\\
      \ &\le \  \exp\left( -\Omega(m) \right)
  \end{align}
  where the last inequality is given by $k\ge m/2$.
  Since our $m\ge \Omega(\log(1/\delta_{out}))$ is sufficiently large, we have that the probability $s_i\id{in}$ is decoded incorrectly is at most $\delta_{out}/10$.
  If we let $Y_j\id{i}$ denote the $Y_j$ corresponding to inner codeword $c_i\id{in}$, the events $E_i$ given by $\sum_j Y_j\id{i}\ge \delta_{in}m$ are independent.
  By concentration bounds on the events $E_i$, we conclude the probability that there are at least $\delta_{out}n/9$ incorrectly decoded inner codewords that are not already affected by spurious buffers and neighboring deleted buffers is $2^{-\Omega(n)}$. 
  
  Our aim is to show that the number of spurious buffers, deleted buffers, and inner decoding failures is small with high probability.
  So far, we have shown that, with high probability, assuming a codeword is not already affected by spurious buffers and neighboring deleted buffers, the number of inner decoding failures is small.
  We now turn to showing the number of spurious buffers is likely to be small.

  A spurious buffer appears inside an inner codeword if many consecutive runs of 1s are type-0.
  A spurious buffer requires at least one of the following: (i) a codeword contains a sequence of at least $\eta m/5$ consecutive type-0 runs of 1s, (ii) a codeword contains a sequence of $\ell\le \eta m/5$ consecutive type-0 runs of 1s, such that, for the $\ell+1$ consecutive runs of 0s neighboring these type-0 runs of 1s, their image under $f\id{del}\circ f\id{dup}$ has at least $0.5\eta m$ 0s.
  We show both happen with low probability within a codeword.
  
  A set of $\ell$ consecutive type-0 runs of 1s occurs with probability at most $10^{-10\ell}$.
  Thus the probability an inner codeword has a sequence of at least $\eta m/5$ consecutive type-0 runs of 1s is at most $m^2\cdot 10^{-10\eta m/5} = \exp(-\Omega(\eta m))$.
  Now assume that in an inner codeword, each set of consecutive type-0 runs of 1s has size at most $\eta m/5$.
  Each set of $\ell$ consecutive type-0 runs of 1s merges $\ell+1$ consecutive runs of 0s in $c$, so that they appear as a single longer run in $s$.
  The sum of the lengths of these $\ell+1$ runs is some number $\ell^*$ that is at most $2\ell+2$.
  The number of bits in the image of these runs of $c_i\id{in}$ under $f\id{del}\circ f\id{dup}$ is distributed as $\Binomial( \ell^* \ceil{B/(1-p)},1-p)$.
  This has expectation $\ell^* B\le 0.41B\eta m$, so by concentration bounds, the probability this run of $s$ has length at least $0.5B\eta m$, i.e. is interpreted as a decoding buffer, is at most $\exp(-\Omega(\eta m))$.
  Hence, conditioned on each set of consecutive type-0 runs of 1s having size at most $\eta m/5$, the probability of having no spurious buffers in a codeword is at least $1-\exp(-\Omega(\eta m))$.
  Thus the overall probability there are no spurious buffers a given inner codeword is at least $(1-\exp(-\Omega(\eta m))(1-\exp(-\Omega(\eta m))) = 1-\exp(-\Omega(\eta m))$.
  Since each inner codeword contains at most $m$ candidate spurious buffers (one for each type-0 run of 1s), the expected number of spurious buffers in an inner codeword is thus at most $m\cdot \exp(-\Omega(\eta m))$.
  By our choice of $m\ge\Omega(\log(1/\delta_{out})/\eta)$, this is at most $\delta_{out}/10$.
  The occurrence of conditions (i) and (ii) above are independent between buffers. 
  The total number of spurious buffers thus is bounded by the sum of $n$ independent random variables each with expectation at most $\delta_{out}/10$.
  By concentration bounds, the probability that there are at least $\delta_{out}n/9$ spurious buffers is $2^{-\Omega(n)}$. 

  A deleted buffer occurs only when the image of the $\ceil{\eta m}$ 0s in a buffer under $f\id{del}\circ f\id{dup}$ is at most $B\eta m/2$.
  The number of such bits is distributed as $\Binomial(\ceil{\eta m}\cdot\ceil{B/(1-p)},1-p)$.
  Thus, for some constant $\zeta$ independent on $m$, each buffer is deleted with probability  $\exp(-\zeta\cdot B\eta m) < \delta_{out}/10$ by our choice of $m\ge \Omega(\log(1/\delta_{out})/\eta)$.
  The events of a buffer receiving too many deletions are independent across buffers.
  By concentration bounds, the probability that there are at least $\delta_{out}n/9$ deleted buffers is thus $2^{-\Omega(n)}$.

  Each inner decoding failure, spurious buffer, and deleted buffer increases the distance $\Delta_{i/d}(c_i\id{out},s_i\id{out})$ by at most $3$: each inner decoding failure causes up to 1 insertion and 1 deletion; each spurious buffer causes up to 1 deletion and 2 insertions; and each deleted buffer causes up to 2 deletions and 1 insertion. 
  Our message is decoded incorrectly if $\Delta_{i/d}(c_i\id{out},s_i\id{out}) > \delta_{out}n$.
  Thus, there is a decoding error in the outer code only if at least one of (i) the number of incorrectly decoded inner codewords without spurious buffers or neighboring deleted buffers, (ii) the number of spurious buffers, or (iii) the number of deleted buffers is at least $\delta_{out}n/9$. However, by the above arguments, each is greater than $\delta_{out}n/9$ with probability $2^{-\Omega(n)}$, so there is a decoding error with probability $2^{-\Omega(n)}$.
  This concludes the proof of Theorem~\ref{thm:random}.
  \qedhere

\subsection{Possible Alternative Constructions}
\label{sec:5-2}

As mentioned in the introduction, Drinea and Mitzenmacher \cite{DrineaM06,DrineaM07} proved that the capacity of the BDC$_p$ is at least $(1-p)/9$.
However, their proof is non-explicit and they do not provide an efficient decoding algorithm.

One might hope that it is possible to use Drinea and Mitzenmacher's construction as a black box.
We could follow the approach in this paper, concatenating an outer code given by \cite{HaeuplerS17} with the rate $(1-p)/9$ random-deletion-correcting code as a black box inner code.
The complexity of the Drinea and Mitzenmacher's so-called \emph{jigsaw decoding} is not apparent from \cite{DrineaM07}.
However, the inner code has constant length, so construction, encoding, and decoding would be constant time.
Thus, the efficiency of the inner code would not affect the asymptotic runtime. 

  The main issue with this approach is that, while the inner code can tolerate random deletions with probability $p$, inner codeword bits are \emph{not} deleted in the concatenated construction according to a BDC$_p$; the 0 bits closer to the buffers between the inner codewords are deleted with higher probability because they might be ``merged'' with a buffer.
  For example, if an inner codeword is $\ab{101111}$, then because the codeword is surrounded by buffers of 0s, deleting the leftmost 1 effectively deletes two bits because the 0 is interpreted as part of the buffer.
  While this may not be a significant issue because the distributions of deletions in this deletion process and BDC$_p$ are quite similar, much more care would be needed to prove correctness.
  
  Our construction does not run into this issue, because our transmitted codewords tend to have \emph{many} 1s on each end of the inner codewords.
  In particular, each inner codeword of $C_{in}$ has either 1 or 11 on each end, so after the Duplication step each inner codeword has $\ceil{B/(1-p)}$ or $2\ceil{B/(1-p)}$ 1s on each end.
  The 1s on the boundary of the inner codeword will all be deleted with probability $\approx\exp(-B)$, which is small.
  Thus, in our construction, it is far more unlikely that bits are merged with the neighboring decoding buffer, than if we were to use a general inner code construction.
  Furthermore, we believe our construction based on bit duplication of a worst-case deletion correcting code is conceptually simpler than appealing to an existential code.

  As a remark, we presented a construction with rate $\frac{(1-p)}{120}$, but using a randomized encoding we can improve the constant from 1/120 to 1/60.
  We can modify our construction so that, during the Duplication step of decoding, instead of replacing each bit of $c\id{cat}$ with a fixed number $\ceil{B/(1-p)}$ copies of itself, we instead replaced each bit independently with $\Poisson(B/(1-p))$ copies of itself.
  Then the image of a run $r_j$ under duplication and deletion is distributed as $\Poisson(B)$, which is independent of $p$.
  Because we do not have a dependence on $p$, we can tighten our bounding in \eqref{eq:random-cher-1} and \eqref{eq:random-cher-2}.
  To obtain a rate of $(1-p)/60$, we can take $B=28.12$ and set $B^*=40$, where $B^*$ is the threshold after which runs are decoded as two bits instead of one bit in the Deduplication step.
  The disadvantage of this approach is that we require our encoding to be randomized, whereas the construction presented above uses deterministic encoding.

%%%%%%%%%%%%%%%%%%%%%%%%%%%%%%%%%%%%%%%%%%%%%%%%%%%%%%%%%%%%%%%%%%%%%%%%%%%%%%%%%%%%%%%%%%%%%%%%%%%%%%%%%%%%%%%%%%%%%%%%%%%%%%%%%%%%%%%%%%%%%%%%%%%%%%%%%%%%%%%%
%%%%%%%%%%%%%%%%%%%%%%%%%%%%%%%%%%%%%%%%%%%%%%%%%%%%%%%%%%%%%%%%%%%%%%%%%%%%%%%%%%%%%%%%%%%%%%%%%%%%%%%%%%%%%%%%%%%%%%%%%%%%%%%%%%%%%%%%%%%%%%%%%%%%%%%%%%%%%%%%

\section{Future work and open questions}

  A lemma due to Levenshtein \cite{Levenshtein1966} states that a code $C$ can decode against $pn$ adversarial deletions if and only if it can decode against $pn$ adversarial insertions and deletions. While this does not automatically preserve the efficiency of the decoding algorithms, all the recent efficient constructions of codes for worst-case deletions also extend to efficient constructions  with similar parameters for recovering from insertions and deletions~\cite{BrakensiekGZ16, GuruswamiL2016}.

  In the random error model, decoding deletions, insertions, and insertions and deletions are not the same.
  Indeed, it is not even clear how to define random insertions.
  One could define insertions and deletions via the Poisson repeat channel  where each bit is replaced with a Poisson many copies of itself (see \cite{DrineaM06,Mitzenmacher2009}).
  However, random insertions do not seem to share the similarity to random deletions that adversarial deletions share with adversarial insertions; we can decode against arbitrarily large Poisson duplication rates, whereas for codes of block length $n$ we can decode against a maximum of $n$ adversarial insertions or deletions \cite{DrineaM07}.
  Alternatively one can consider a model of random insertions and deletions where, for every bit, the bit is deleted with a fixed probability $p_1$, a bit is inserted after it with a fixed probability $p_2$, or it is transmitted unmodified with probability $1-p_1-p_2$  \cite{VenkataramananTR2013}.
One could also investigate settings involving memoryless insertions, deletions, and substitutions \cite{MercierTL2012}.

There remain a number of open questions even concerning codes for deletions only.
Here are a few highlighted by this work.
\begin{enumerate}
  \item Can we close the gap between $\sqrt2-1$ and $\frac{1}{2}$ on the maximum correctable fraction of  adversarial deletions?
  \item Can we construct efficiently decodable codes for the binary deletion channel with better rate, perhaps reaching or beating the best known existential capacity lower bound of $(1-p)/9$?
  \item Can we construct efficient codes for the binary deletion channel with rate $1-O(h(p))$ for $p \to 0$? 

\end{enumerate}

\section{Acknowledgements}

We thank anonymous referees for helpful feedback on earlier versions of this paper.

\bibliographystyle{plainurl}
\bibliography{GuruswamiLiRandomDeletionTRANSIT}

%%%%%%%%%%%%%%%%%%%%%%%%%%%%%%%%%%%%%%%%%%%%%%%%%%%%%%%%%%%%%%%%%%%%%%%%%%%%%%%%%%%%%%%%%%%%%%%%%%%%%%%%%%%%%%%%%%%%%%%%%%%%%%%%%%%%%%%%%%%%%%%%%%%%%%%%%%%%%%%%
%%%%%%%%%%%%%%%%%%%%%%%%%%%%%%%%%%%%%%%%%%%%%%%%%%%%%%%%%%%%%%%%%%%%%%%%%%%%%%%%%%%%%%%%%%%%%%%%%%%%%%%%%%%%%%%%%%%%%%%%%%%%%%%%%%%%%%%%%%%%%%%%%%%%%%%%%%%%%%%%

\appendix

%%%%%%%%%%%%%%%%%%%%%%%%%%%%%%%%%%%%%%%%%%%%%%%%%%%%%%%%%%%%%%%%%%%%%%%%%%%%%%%%%%%%%%%%%%%%%%%%%%%%%%%%%%%%%%%%%%%%%%%%%%%%%%%%%%%%%%%%%%%%%%%%%%%%%%%%%%%%%%%%
%%%%%%%%%%%%%%%%%%%%%%%%%%%%%%%%%%%%%%%%%%%%%%%%%%%%%%%%%%%%%%%%%%%%%%%%%%%%%%%%%%%%%%%%%%%%%%%%%%%%%%%%%%%%%%%%%%%%%%%%%%%%%%%%%%%%%%%%%%%%%%%%%%%%%%%%%%%%%%%%

  \section{Proof of Lemma~\ref{lem:chernoff-2}}
    \label{app:A}
    \begin{proof}
      For each $i$, we can find a random variable $B_i$ such that $B_i\ge A_i$ always, $B_i$ takes values in $[0,\beta]$, and $\E[B_i] = \alpha$.
      In the second part of Lemma~\ref{lem:chernoff}, the inequality \eqref{eq:chernoff-2} can weakened to a useful form:
      \begin{equation}
        \Pr[A\ge (1+\delta)\E[A]] 
        \ \le \ \left( \frac{e^\delta}{(1+\delta)^{1+\delta}}\right)^{\E[A]}
        \ \le \ \exp\left(-\delta^2\E[A]/3\right). 
      \end{equation}
      To obtain this, use the approximation that $(1+\delta)\ge e^{\delta-\delta^2/3}$ for $\delta\in[0,1]$, so that
      \begin{align}
        \frac{e^\delta}{(1+\delta)^{1+\delta}}
        \ \le \ \frac{e^{\delta}}{e^{(\delta-\delta^2/3)(1+\delta)}}
        \ = \ \exp\left(\delta^2\cdot \frac{-2+\delta}{3}\right)
        \ \le \ \exp\left( -\delta^2/3 \right).
      \label{}
      \end{align}
      Applying this weaker form of Lemma~\ref{lem:chernoff}, we have
      \begin{align}
        \Pr\left[ \sum_{i=1}^n A_i \ge n\gamma \right]
        \ &\le \ \Pr\left[ \sum_{i=1}^n B_i \ge n\gamma \right]  \nonumber\\
        \ &\le \ \Pr\left[ \sum_{i=1}^n \frac{B_i}{\beta} \ge \left( 1 + \left( \frac{\gamma - \alpha}{\alpha} \right) \right)\frac{n\alpha}{\beta} \right]  \nonumber\\
        \ &\le \ \exp\left( - \frac{\left(\frac{\gamma-\alpha}{\alpha}\right)^2\cdot \frac{n\alpha}{\beta}}{3} \right)  \nonumber\\
        \ &= \  \exp\left( -\frac{(\gamma - \alpha)^2n}{3\alpha\beta} \right).
        \qedhere \nonumber
      \end{align}
      Here, we applied Lemma~\ref{lem:chernoff} to the variables $\frac{B_i}{\beta}$ by setting $\delta = \frac{\gamma-\alpha}{\alpha}\in[0,1]$, and using that the expectation of $\sum_{i=1}^{n} \frac{B_i}{\beta}$ is $\frac{n\alpha}{\beta}$.
    \end{proof}

\end{document}